%% file: main.tex
\documentclass[12pt]{amsart}

\makeatletter
\def\l@subsection{\@tocline{2}{0pt}{2.5pc}{5pc}{}}
\def\l@subsubsection{\@tocline{2}{0pt}{5pc}{7.5pc}{}}
\makeatother

\usepackage{fullpage,amssymb}
\usepackage{mathrsfs}
\usepackage{color}
\usepackage{algorithmicx}
\usepackage{algpseudocode}
\usepackage{tikz-cd}
\usepackage{tikz}
\usetikzlibrary{matrix,decorations.pathreplacing}
\usepackage[T1]{fontenc}
\usepackage{amsbsy}
\usepackage{bm}
\usepackage{hyperref}

\usepackage{color}   
\hypersetup{
	hypertexnames=false,
    colorlinks=true, 
    linktoc=all,     
    linkcolor=blue,  
    urlcolor=blue,
    citecolor=blue,
    anchorcolor=green,
    breaklinks=true,
}

\usepackage[vcentermath]{youngtab}

\newtheorem{theorem}{Theorem}[section]
\newtheorem{conjecture}[theorem]{Conjecture}
\newtheorem{corollary}[theorem]{Corollary}

\newtheorem{proposition}[theorem]{Proposition}

\theoremstyle{definition}
\newtheorem{definition}[theorem]{Definition}
\newtheorem{remark}[theorem]{Remark}

\newtheorem{problem}[theorem]{Problem}

\newcommand{\C} {\mathbb{C}}
\newcommand{\Z} {\mathbb{Z}}
\newcommand{\F} {\mathbb{F}}
\newcommand{\N} {\mathbb{N}}
\newcommand{\Q} {\mathbb{Q}}
\newcommand{\R} {\mathbb{R}}

\newcommand{\GL} {\operatorname{GL}}
\newcommand{\T}{\operatorname{T}}
\newcommand{\ST} {\operatorname{ST}}
\newcommand{\SL} {\operatorname{SL}}
\newcommand{\Sym}{\operatorname{Sym}}
\newcommand{\mat}{\operatorname{M}}

\DeclareMathOperator{\per}{Per}
\DeclareMathOperator{\sgn}{sgn}
\DeclareMathOperator{\Pf}{Pf}

\newcommand{\poly} {{\rm poly}}

\newcommand{\OO} {\mathcal{O}}
\newcommand{\bcO}{\overline{\OO}}

\newcommand{\cC} {\mathcal{C}}

\newcommand{\cF}{\mathcal{F}}
\newcommand{\cP}{\mathcal{P}}
\newcommand{\cM}{\mathcal{M}}
\newcommand{\cN}{\mathcal{N}}

\newcommand{\tensor}{{\textstyle\bigotimes}}
\newcommand{\la}{\lambda}

\newcommand{\VNP}{\textup{VNP}}

\newcommand{\bz}{\textbf{z}}
\newcommand{\by}{\textbf{y}}

\newcommand{\ba}{\textbf{a}}
\newcommand{\bb}{\textbf{b}}
\newcommand{\bc}{\textbf{c}}

\newcommand{\bx}{\textbf{x}}

\newcommand{\bbm} {\textbf{m}}

\newcommand{\supp} {{\rm supp}}

\newcommand{\Rnote}[1]{\begin{quote}{\color{blue}{{\sf Rafael's Note:} {\sl{#1}}}} \end{quote}}

\title{Search problems in algebraic complexity, GCT, and hardness of generators for invariant rings}

\author{Ankit Garg}
\address{Microsoft Research India, Bangalore}
\email{garga@microsoft.com}

\author{Christian Ikenmeyer}
\thanks{Christian's research was supported by DFG grant IK 116/2-1.}
\address{University of Liverpool}
\email{christian.ikenmeyer@liverpool.ac.uk}

\author{Visu Makam}
\thanks{Visu's research was supported by NSF grant No. DMS -1638352 and NSF grant No. CCF-1412958.}
\address{School of Mathematics, Institute for Advanced Study, Princeton}
\email{visu@ias.edu}

\author{Rafael Oliveira}
\thanks{Part of this work was done while Rafael was at the Simons Institute and at University of Toronto.}
\address{University of Waterloo}
\email{rafael@uwaterloo.edu}

\author{Michael Walter}
\thanks{Michael's research was supported by NWO Veni grant 680-47-459.}
\address{Korteweg-de Vries Institute for Mathematics, Institute for Theoretical Physics, Institute for Logic, Language \& Computation, University of Amsterdam}
\email{m.walter@uva.nl}

\author{Avi Wigderson}
\thanks{Avi's research is supported by NSF grant No. CCF-1412958 and NSF grant CCF-1900460.}
\address{School of Mathematics, Institute for Advanced Study, Princeton}
\email{avi@ias.edu}

\begin{document}

\begin{abstract}
We consider the problem of computing succinct encodings of lists of generators for invariant rings for group actions. Mulmuley conjectured that there are always polynomial sized such encodings for invariant rings of $\SL_n(\C)$-representations. We provide simple examples that disprove this conjecture (under standard complexity assumptions).

We develop a general framework, denoted \emph{algebraic circuit search problems}, that captures many important problems in algebraic complexity and computational invariant theory.
This framework encompasses various proof systems in proof complexity and some of the central problems in invariant theory as exposed by the Geometric Complexity Theory (GCT) program, including the aforementioned problem of computing succinct encodings for generators for invariant rings.

\end{abstract}

\maketitle

\section{Introduction}

In complexity theory, one often encounters problems that ask for an efficiently computable collection of functions/polynomials satisfying a certain property.
Once we are faced with such problems, two natural questions are: how do we represent the property?
And how do we encode the required functions? The answer to these will depend on the context and use.
We will first define the informal notion of an algebraic circuit search problem, and then give illustrative examples.


\begin{definition}[Algebraic circuit search problems]
Given an input (of size $n$), construct an algebraic circuit in a complexity class $\cC$ with $k(n)$-inputs and $m(n)$-outputs such that the polynomials they compute satisfy a desirable property $\cP$.
\end{definition}

Let us illustrate this definition in the context of algebraic proof complexity: in Null\-stellen\-satz-based
proof systems, one is given a set of multivariate polynomials $g_1,\dots,g_r$ over
an algebraically closed field $\F$ and in variables $\bx = (x_1, \ldots, x_n)$, and one wants to
decide whether the system $g_1(\bx) = g_2(\bx) = \ldots = g_r(\bx) = 0$ has a
solution over $\F$.
A fundamental result of Hilbert tells us that the system has no solution if and only if there is
a set of polynomials $\{f_i\}_{i=1}^r$ such that $\sum_i f_i g_i = 1$. This brings us to the Ideal Proof System~\cite{GP14}:
\begin{itemize}
\item \textbf{Ideal Proof System (IPS)}: Given a collection of polynomials $g_1(\bx), g_2(\bx), \dots, g_r(\bx)$, we ask to construct a polynomial sized circuit $C$ with $n + r$ inputs. The desirable property $\cP$ is that $C(x_1,\dots,x_n,g_1(\bx),\dots,g_r(\bx)) = 1$ and that $C(x_1,\dots,x_n,0,\dots,0) = 0$. It is not so hard to see these conditions will give us a linear combination of the form $\sum_i f_i g_i = 1$ as required.
\end{itemize}

\noindent
In~\cite{GP14} the authors show that super-polynomial lower bounds in this proof system imply
algebraic circuit lower bounds (i.e., ${\rm VP \neq VNP}$), which remains a long standing
open problem in complexity theory. Another important point to make is that an instance of 3-SAT, say $\phi$, can be encoded as a collection of polynomials $\{g_i\}$ such that $\phi$ is satisfiable if and only if the $\{g_i\}$ have a common solution. In other words,  $\phi$ is unsatisfiable if and only if $\exists$ polynomials $f_i$ such that $\sum_i f_i g_i = 1$. This converts a co-NP complete problem (unsatisfiability of 3-SAT, called co-3-SAT) into an algebraic circuit search problem of the IPS form described above. The existence of a polynomial sized circuit as demanded by the IPS proof system would mean the existence of $f_i$ with small circuits. But that would mean that co-3-SAT is in NP, thus proving NP = co-NP.

Other important examples of algebraic search problems in proof complexity (with different desirable properties)
are the original Nullstellensatz proof system, Polynomial Calculus, and the Positivstellensatz\footnote{In this case our field is the real numbers, which is \emph{not} algebraically closed.} for sum of squares (SOS) proofs.
For more on these systems we refer the reader to~\cite[Chapter 6]{K19}.

\begin{remark}
An analogous notion of a ``boolean circuit search problem'' can also be introduced in the boolean setting. Also here, important problems such as the construction of pseudorandom generators and the construction of extractors can be captured as boolean circuit search problems.
\end{remark}

\subsection{Geometric Complexity Theory}
The GCT program was proposed by Mulmuley and Sohoni (see \cite{MS01,MS08}) as an approach
(via representation theory and algebraic geometry) to the VP vs.\ VNP problem. While there have been some negative results\footnote{The results of B\"{u}rgisser, Ikenmeyer and Panova, which show that occurrence obstructions cannot give a super-polynomial lower bound
on the determinantal complexity of the permanent polynomial (see \cite{BIP}).} in recent years regarding the techniques one can use towards this program, these results do not disrupt the core framework of the GCT program. Instead, these results indicate the difficulty of the problem from the viewpoint of algebraic combinatorics, and have identified new directions of research in asymptotic algebraic combinatorics. In \cite{GCTV}, Mulmuley views the VP vs.\ VNP problem through the lens of computational invariant theory, and identifies important and interesting problems in computational invariant theory that form a path towards resolving the VP vs.\ VNP problem. These include several conjectures, some of which fit into the framework of algebraic circuit search problems, and have important connections and
consequences to problems in optimization, algebraic complexity, non-commutative computation,
functional analysis and quantum information theory (see \cite{GGOW16, GGOW18, BFGOWW18}).
We therefore believe that a better understanding of algebraic circuit search problems will likely result in fundamental advances in the aforementioned areas.
Some evidence for these conjectures has emerged over the past few years as they have been established for special cases
(see, for example, \cite{GCTV,GGOW16,IQS2, FS, DM, DMOC,DM-arbchar}).


Let us briefly mention an important algebraic circuit search problem and one that will be central to this paper: given a group action, describe a set of generators for the invariant ring (we will elaborate on invariant theory in a subsequent section). Unfortunately, the number of generators for an invariant ring is usually exponential (in the input size of the description of the action). So, in order to get a computational handle on them, Mulmuley suggests in \cite{GCTV} that we should look for a {\em succinct encoding} (defined below in Definition~\ref{def:succ.enc}) using some auxiliary variables. One amazing feature of such a succinct encoding is that it would immediately give efficient randomized algorithms for null cone membership and the orbit closure intersection problems which can then be derandomized in some cases (see, e.g., \cite{FS,IQS2,DMOC}). We will define these problems in a subsequent section, but here we are content to say that many important algorithmic problems such as graph isomorphism, bipartite matching, (non-commutative) rational identity testing, tensor scaling and a form of quantum entanglement distillation are all specific instances (or arise in the study) of null cone membership and orbit closure intersection problems.

Mulmuley conjectures (\cite[Conjecture~5.3]{GCTV}\footnote{In the conjecture, the group is specified to be $\SL_n(\C)$, which was done for the purpose of accessibility and brevity, but it is natural to ask this problem for general connected reductive groups. We will discuss this further in a later section.}) the existence of polynomial sized succinct encodings for generators of invariant rings. The main goal of this paper is to (conditionally) disprove this conjecture. More precisely we give an example of an invariant ring (for a torus action) where the existence of such a circuit would imply a polynomial time algorithm for the $3$D-matching problem, which is well known to be NP-hard. We also give another example (where the group is $\SL_n(\C)$) where the existence of such a circuit would imply VP $\neq$ VNP. Further, the nature of the latter example makes it clear that no simple modification of this conjecture can hold.


The rest of this section will proceed as follows. We first give a brief introduction to invariant theory. Then, we discuss the algebraic search problems of interest in computational invariant theory, followed by the precise statements of our main results. Finally, we discuss some open problems and future directions.

\subsection{Invariant Theory}\label{subsec:inv theo}

Invariant theory is the study of symmetries, captured by group actions on vector spaces (more generally, algebraic varieties), by focusing on the functions (usually, polynomials) that are left \emph{invariant} under these actions. It is a rich mathematical field in which computational methods are sought and well developed (see \cite{DK, Sturmfels}). While significant advances have been made on computational problems involving invariant theory, most algorithms are based on Gr\"obner bases techniques, and hence still require exponential time (or longer).

The basic setting is that of a continuous group%
\footnote{In general, the theory works whenever the group is connected, algebraic and reductive.
However in this paper, we will deal with very simple groups.}
$G$ acting (linearly) on a finite-dimensional vector space~$V = \C^m$.

An {\em action} (also called a {\em representation}) of a group $G \subseteq \GL_n(\C)$ on an $m$-dimensional complex vector space $V$ is a group homomorphism~$\pi\colon G \rightarrow \GL_m(\C)$, that is, an association of an invertible~$m \times m$ matrix~$\pi(g)$ for every group element~$g\in G$, satisfying $\pi(g_1 g_2) = \pi(g_1) \pi(g_2)$ for all~$g_1,g_2\in G$.
To be precise, a group element $g \in G$ acts on a vector $v \in V$ by the linear transformation $\pi(g)$, and in this paper we will be dealing with algebraic actions, that is, the entries of the matrix $\pi(g)$ will be rational functions in the entries of the matrix $g$.
We will write $g \cdot v = \pi(g) v$. Invariant theory is nicest when the underlying field is $\C$ and the group $G$ is either finite, the general linear group $\GL_n(\C)$, the special linear group $\SL_n(\C)$, or a direct product of these groups and their diagonal subgroups. We denote by $\C[V]$ the ring of polynomial functions on $V$.

\textbf{Invariant Polynomials}:
Invariant polynomials are precisely those which cannot distinguish between a vector $v$ and a translate of it by an element of the group, i.e., $g \cdot v$. In other words, a polynomial function $f \in \C[V]$ is called invariant if $f(g \cdot v) = f(v)$ for all $v\in V$ and $g \in G$. Equivalently, invariant polynomials are polynomial functions on $V$ which are left invariant by the action of $G$. More precisely, the action of $G$ on $V$ gives an induced action of $G$ on $\C[V]$, the space of polynomial functions on $V$.
For a polynomial function $p$ on $V$, the group element $g \in G$ sends it to the function $g \cdot p$ which is defined by the formula $(g \cdot p) (v) = p(g^{-1} \cdot v)$ for $v \in V$.
Then, a polynomial function is invariant if and only if~$g \cdot p = p$ for all~$g \in G$.
A set~$\{f_i\}_{i \in I}$ of invariant polynomials is called a \emph{generating set} if any other invariant polynomial can be written as a polynomial in the $f_i$'s.
Two simple and illustrative examples are
\begin{itemize}
\item The symmetric group $G = \mathcal{S}_n$ acts on $V = \C^n$ by permuting the coordinates. In this case, the invariant polynomials are \emph{symmetric} polynomials, and the $n$ elementary symmetric polynomials form a generating set (a result that dates back to Newton).

\item The group $G = \SL_n(\C) \times \SL_n(\C)$ acts on $V = \mat_n(\C)$ by a change of bases of the rows and columns, namely left-right multiplication: that is, the action of $(A,B)$ sends~$X$ to $A X B^T$.
Here, $\det(X)$ is an invariant polynomial and in fact every invariant polynomial must be a univariate polynomial in $\det(X)$. In other words, $\det(X)$ generates the invariant ring.
\end{itemize}
The above phenomenon that the ring of invariant of polynomials (denoted by $\C[V]^G$) is generated by a finite number of invariant polynomials is not a coincidence. The \emph{finite generation theorem} due to Hilbert \cite{Hilb1, Hilb2} states that, for a large class of groups (including the groups mentioned above), the invariant ring must be finitely generated. These two papers of Hilbert are highly influential and laid the foundations of modern commutative algebra and algebraic geometry. In particular, ``finite basis theorem'' and ``Nullstellansatz'' were proved as ``lemmas'' on the way towards proving the finite generation theorem!

\textbf{Orbits and Orbit Closures}:
The \emph{orbit} of a vector $v\in V$, denoted by $\OO_v$, is the set of all vectors obtained by the action of $G$ on $v$.
The \emph{orbit closure} of $v$, denoted by $\overline{\OO}_v$, is the closure of the orbit $\OO_v$ in the Euclidean topology.\footnote{It turns out mathematically more natural to look at closure under the Zariski topology. However, for the group actions we study, the Euclidean and Zariski closures match by a theorem due to Mumford \cite{Mum65}.}
For actions of continuous groups, such as $\GL_n(\C)$, it is more natural to look at orbit closures. The \emph{null cone} for a group action is the set of all vectors which behave like the $0$ vector i.e. the $0$ vector lies in their orbit closure.
Many fundamental problems in theoretical computer science (and many more across mathematics) can be phrased as questions about orbits and orbit closures. Here are some familiar examples:
\begin{itemize}
\item Graph isomorphism problem can be phrased as checking if the orbits of two graphs are the same or not, under the action of the symmetric group permuting the vertices.
\item Geometric complexity theory (GCT) \cite{MS01} formulates a variant of the VP vs.\ VNP question as checking if the (padded) permanent lies in the orbit closure of the determinant (of an appropriate size), under the action of the general linear group on polynomials induced by its natural linear action on the variables.
\item Border rank (a variant of tensor rank) of a $3$-tensor can be formulated as the minimum dimension such that the (padded) tensor lies in the orbit closure of the unit tensor, under the natural action of $\GL_r(\C) \times \GL_r(\C) \times \GL_r(\C)$. In particular, this captures the complexity of matrix multiplication.
\end{itemize}

\subsection{Computational invariant theory, Mulmuley's problems and conjectures}
From its origins in the 19th century, the subject of classical invariant theory has been computational in nature -- one of its central goals is explicit descriptions of generators of invariant rings, their relations, etc. With the more recent advent of the theory of computation, it is only natural to ask for the complexity of these descriptions. The influence of complexity theory has taken an important role in invariant theory as a consequence of the connections to fundamental problems such as VP vs.\ VNP that were uncovered as part of the GCT program by Mulmuley in \cite{GCTV}. In \cite{GCTV}, Mulmuley considers the computational complexity of various invariant theoretic problems. Let $G$ be a group acting on $V$.

\begin{enumerate}
\item (\textbf{Generators}) Output a list of polynomials that generate the invariant ring $\C[V]^G$.
\item (\textbf{NNL}) Output a list of polynomials $f_1,\dots,f_r$, such that each $f_i$ is a homogeneous polynomial and the invariant ring $\C[V]^G$ is integral over $\C[f_1,\dots,f_r]$.\footnote{This is equivalent to the condition that the zero locus of $f_1,\dots,f_r$ is precisely the null cone.}
\item (\textbf{Orbit closure intersection}) Given two elements of the vector space, do their
orbit closures intersect?
\item (\textbf{Null cone membership}) Given an element of the vector space, does the $0$
vector lie in its orbit closure?
\end{enumerate}

\noindent
Let us point out straight away that Generators and NNL (Noether Normalization Lemma) are both algebraic circuit search problems (we will define Generators as an algebraic circuit search problem more precisely below).
Orbit closure intersection and Null cone membership are not algebraic circuit search problems, but are related to Generators and NNL in a way that will become clear in a later discussion.
We will not get into the details of how the group is given and how the group action is described.
It turns out that even for simple groups and group actions, these problems turn out to be interesting.
They have been long studied and many algorithms have been developed in the invariant theory community \cite{DK, Sturmfels}.
Mulmuley \cite{GCTV} introduced these problems to theoretical computer science with the hope of making progress on the polynomial identity testing (PIT) problem. Before describing the main conjectures in Mulmuley's paper, let us see what it even means to output a list of generating polynomials for an invariant ring. Typically the number of generating polynomials can be exponential in the dimension of the group and the vector space. To get around this issue, Mulmuley introduced the following notion of a \emph{succinct encoding} of the generators of an invariant ring (which in fact applies to any collection of polynomials).

\begin{definition}[\textbf{Succinct encoding of generators}] \label{def:succ.enc}
	Fix an action of a group $G$ on a vector space $V = \C^m$. We say that an arithmetic circuit
	$\cC(x_1,\ldots, x_m,y_1,\ldots, y_r)$ \emph{succinctly encodes} the generators of the invariant
	ring if the set of polynomials formed by evaluating the $y$-variables,
	$\{\cC(x_1,\ldots, x_m, \alpha_1,\ldots, \alpha_r)\}_{\alpha_1,\ldots, \alpha_r \in \C}$,
  is a generating set for the invariant ring~$\C[V]^G$.
\end{definition}

\begin{remark}
	The \emph{size} of a succinct encoding as defined above is given by the size of the circuit
	$\cC(x_1, \ldots, x_m, y_1, \ldots, y_r)$, which is measured by the bit complexity of the
	constants used in the computation of $\cC$ as well as the number of gates of the computation
	graph of $\cC$.
  In particular, this means that all constants used in the computation of $\cC$ are rationals.
\end{remark}

The above notion of a succinct encoding motivates us to define the following algebraic search problem.

\begin{problem} [\textbf{Generators}] \label{prob:generators}
Let $G$ be a group of dimension $n$ and that acts algebraically on an $m$-dimensional vector space $V$ by linear transformations. Output a ${\rm poly}(n,m)$ sized circuit $\cC(x_1,\ldots, x_m,y_1,\ldots, y_r)$ such that the polynomials $\{\cC(x_1,\ldots, x_m, \alpha_1,\ldots, \alpha_r)\}_{\alpha_1,\ldots, \alpha_r \in \C}$ form a generating set for the invariant ring $\C[V]^G$.
In other words, the problem asks to output a ${\rm poly}(n,m)$ sized succinct encoding for the generators of $\C[V]^G$.
\end{problem}

\begin{conjecture}[Mulmuley]\label{con:generators}
	In the case that $G$ is a connected reductive algebraic group\footnote{We have not defined what
	a connected reductive algebraic group is. One should think of simple groups like the
	general linear group $\GL_n(\C)$, the special linear group $\SL_n(\C)$, or a direct product
	of these groups and their diagonal subgroups.}, Problem~\ref{prob:generators} has a positive answer.
  That is, there exists a $\poly(n,m)$ sized circuit which succinctly encodes the generators of $\C[V]^G$.
\end{conjecture}

\noindent
Mulmuley requires the circuit family (that succinctly encodes the generators) to be uniformly computable by a polynomial time algorithm, but we will see that even this weaker conjecture is false (under standard complexity assumptions).

In~\cite[Conjecture~5.3]{GCTV}, Mulmuley states the above conjecture for actions of the group $\SL_n(\C)$.
However, it is evident that there is nothing special about $\SL_n(\C)$ with regard to the GCT program and it is natural to state the conjecture in the generality of connected reductive groups.
Let us also note that it was already evident to Mulmuley that one cannot drop the ``connected'' assumption on the group, because the permanent appears as an invariant polynomial for a non-connected reductive group that would disprove the conjecture immediately using a similar line of reasoning to the one we use in the next section (see, e.g., \cite{BIJL:18}).


To understand Mulmuley's motivation for the conjecture, let us see what it means for the problems of orbit closure intersection and null cone membership. By definition, invariant polynomials are constant on the orbits (and thus on orbit closures as well). Thus, if $\bcO_{v_1} \cap \bcO_{v_2} \neq \emptyset$, then $p(v_1) = p(v_2)$ for all invariant polynomials $p \in \C[V]^G$. A remarkable theorem due to Mumford says that the converse is also true for the large class of reductive groups:

\begin{theorem}[\cite{Mum65}]\label{introthm:Mum} Fix an action of a reductive group $G$ on a vector space $V$. Given two vectors $v_1, v_2 \in V$, we have $\bcO_{v_1} \cap \bcO_{v_2} \neq \emptyset$ if and only if $p(v_1) = p(v_2)$ for all $p \in \C[V]^G$.
\end{theorem}

Now suppose one had a succinct encoding $\cC(x_1,\ldots, x_m,y_1,\ldots, y_r)$ for action of a group $G$ on $V = \C^m$. Then because of Mumford's theorem, for two vectors $v_1$ and $v_2$, their orbit closures intersect iff the two polynomials $\cC(v_1(1),\ldots, v_1(m),y_1,\ldots, y_r)$, \linebreak $\cC(v_2(1),\ldots, v_2(m),y_1,\ldots, y_r)$ are identically the same. These are instances of polynomial identity testing (PIT)! Thus if Conjecture \ref{con:generators} were true (and additionally the succinct encoding circuits were polynomial time computable), it immediately gives randomized polynomial time algorithms for the orbit closure intersection and null cone membership problems.
This also gives us a nice family of PIT problems to play with. Perhaps one might hope that solving these PIT instances will result in development of new techniques which might shed a light on the general PIT problem. In fact, for the first few group actions that were studied in this line of work, \emph{simultaneous conjugation} \cite{GCTV, FS} and \emph{left-right action} \cite{GGOW16, IQS2, DM}, for which there are polynomial sized succinct encodings of generators, the null cone membership problems correspond to PIT problems for restricted models of computation: \emph{read-once algebraic branching programs} and \emph{non-commutative formulas with division}\footnote{Actually a stronger model concerning inverses of matrices.}, both of which have been successfully derandomized, see \cite{FS,GGOW16,IQS2}.

\subsection{Our results}

While the truth of Conjecture~\ref{con:generators} would have great implications, we prove that it is false under plausible complexity hypotheses.
We first state our counterexamples (they are very simple, and probably many others exist), and then discuss how a related conjecture may be true and almost as powerful as the original.

For our first counterexample, we analyze a simple (torus) action on $3$-tensors.
Here, $\ST_n(\C)$ denotes the group of $n\times n$ diagonal matrices with determinant $1$.

\begin{theorem}\label{thm:torus-action}
Consider the natural action of $G = \ST_n(\C) \times \ST_n(\C) \times \ST_n(\C)$ on $V = \C^n \otimes \C^n \otimes \C^n$. Then any set of generators for the invariant ring cannot have a polynomial sized (in $n$) succinct encoding, unless ${\rm NP \subseteq P/poly}$.
\end{theorem}

\begin{corollary}
Conjecture~\ref{con:generators} is false, unless ${\rm NP \subseteq P/poly}$.
\end{corollary}

\begin{remark}
As mentioned previously, a primary motivation for succinct encodings of generators is that they imply (randomized) polynomial time algorithms for null cone membership problem. For the action in Theorem~\ref{thm:torus-action}, it is important to note that even though we do not have a succinct encoding for generators, we still have a polynomial time algorithm for null cone membership since once can reduce it to an instance of linear programming.
For a general connection between null cone membership and optimization, see~\cite{burgisser2019towards}.
\end{remark}

For the above counterexample for the torus action, the notion of a succinct encoding is quite crucial to our argument, and it is natural to wonder if tweaking the notion would get rid of the issue. We give another counterexample where it becomes apparent that the precise form of encoding of the generators is not quite as crucial, as we identify an invariant that is hard to compute and is \emph{essential} to any generating set in a sense that we will make precise in Section~\ref{Sec:Ketan}. Moreover, it is an $\SL_n(\C)$-action, which provides a counterexample to the exact formulation of the conjecture in \cite{GCTV}. 

\begin{theorem}\label{MC-cex}
Let $k \geq 2$ be even.
Consider the action of $G = \SL_{2kn}(\C)$ on $V = \tensor^{2k} \C^{2kn}$. Then any set of generators for the invariant ring cannot have a polynomial sized (in $n$) succinct encoding, unless ${\rm VP = VNP}$.
\end{theorem}

\begin{corollary}
Conjecture~\ref{con:generators} is false, unless ${\rm VP = VNP}$.
\end{corollary}







\subsection{Conclusion, open problems and future directions}

We have disproved a conjecture of Mulmuley about the existence of polynomial sized succinct encodings of generators for invariant rings. We want to emphasize that this only serves a first guiding light for Mulmuley's program of understanding the orbit closure intersection problems (and null cone membership problems) and connections to PIT. To solve the orbit closure intersection problems, one does not necessarily need a generating set of generators. This motivates the following definition.

\begin{definition} [Separating set of invariants]
For a group $G$ acting algebraically on a vector space $V$ by linear transformations, a subset $S \subseteq \C[V]^G$ is called a \emph{separating set of invariants} if for all $u,v \in V$ such that $\bcO_u \cap \bcO_v \neq \emptyset$, there exists $f \in S$ such that $f(u) \neq f(v)$.
\end{definition}

\noindent
This leads to a natural algebraic search problem that corresponds to the algorithmic problem of orbit closure intersection.
Mulmuley already suggested that a positive answer to the following search problem would suffice for the purposes of GCT.

\begin{problem} [\textbf{Separators}] \label{Prob:sep}
Let $G$ be a group of dimension $n$ and suppose it acts algebraically on an $m$-dimensional vector space $V$ by linear transformations. Output a ${\rm poly}(n,m)$ sized circuit $\cC(x_1,\ldots, x_m,y_1,\ldots, y_r)$, if one exists, such that the set of polynomials $S = \{\cC(x_1,\ldots, x_m, \alpha_1,\ldots, \alpha_r)\}_{\alpha_1,\ldots, \alpha_r \in \C}$ is a separating set of invariants.
\end{problem}

Similarly, we can define a search problem that corresponds to the algorithmic problem of null cone membership.

\begin{problem} [\textbf{Null cone definers}] \label{Prob:ncd}
Let $G$ be a group of dimension $n$ and suppose $G$ acts algebraically on an $m$-dimensional vector space $V$ by linear transformations. Output a ${\rm poly}(n,m)$ sized circuit $\cC(x_1,\ldots, x_m,y_1,\ldots, y_r)$ with the property that the set $S = \{\cC(x_1,\ldots, x_m, \alpha_1,\ldots, \alpha_r)\}_{\alpha_1,\ldots, \alpha_r \in \C}$ consists of
	invariant polynomials whose zero locus is precisely the null cone $\mathcal{N}_G(V) = \{v \in V\ |\ 0 \in \bcO_v\}$.
\end{problem}


We conclude the introduction with some open open problems:

\begin{enumerate}
\item Are there polynomial sized succinct encodings for separating invariants or, even simpler, invariants defining the null cone? In other words, do we have positive answers to Problems~\ref{Prob:sep} and \ref{Prob:ncd} for connected reductive groups? Perhaps the first non-trivial example is the natural action of $G = \ST_n(\C) \times \ST_n(\C) \times \ST_n(\C)$ on $V = \C^n \otimes \C^n \otimes \C^n$. Here a tensor $T$ is in the null cone iff there exists vectors $x,y,z \in \R^n$ s.t. $x_i + y_j +z_k > 0$ for all $(i,j,k) \in \supp(T)$\footnote{$\supp(T) = \{(i,j,k) \in [n] \times [n] \times [n]: T_{i,j,k} \neq 0\}$.} and $\sum_i x_i = \sum_j y_j = \sum_k z_k = 0$ (by the Hilbert-Mumford criterion). Is there a polynomial sized circuit $\C((z_{i,j,k}), y_1,\ldots, y_r)$ s.t. $\C(T, y_1,\ldots, y_r)$ is identically zero (as a polynomial in the $y$-variables) iff $T$ is in the null cone?
\item For the natural action of $\SL_n(\C) \times \SL_n(\C) \times \SL_n(\C)$ on $V = \C^n \otimes \C^n \otimes \C^n$, it is not even clear if there exists one invariant which has a polynomial sized circuit. Either produce such an invariant or prove that all invariants are hard to compute.
\item Are there polynomial time algorithms for the orbit closure intersection and null cone membership problems? The analytic approach pursued in the papers \cite{GGOW16, burgisser2017alternating,AGLOW18,burgisser2019towards} seems the most promising approach towards getting such algorithms.
\item More broadly, invariant theory is begging for its own complexity theory and connecting it with ours. This includes finding reductions and completeness results, and characterizations/dichotomies about hard/easy actions. An example of a completeness reduction is the reduction from all quiver actions to the simple left-right action \cite{derksen2000semi, domokos2001semi, schofield2001semi, DM, GGOW18}. Also the papers \cite{GCTV,GGOW16,IQS2, FS, DM, DMOC,DM-arbchar}, as well as the current paper, are trying to identify easy and hard problems in invariant theory.
\end{enumerate}

\section{Preliminaries}

\input{prelim}

\section{Hardness of Generators for torus actions}\label{sec:torus}

Let $\C^*$ denote the multiplicative group consisting of all non-zero
complex numbers. A direct product $\T_n = (\C^*)^n$ is called a torus, and
is clearly an abelian group. Tori are important examples of reductive
groups -- any abelian connected reductive group is a torus! It is often the case that it is easier to understand tori in comparison with more general (non-abelian) reductive groups. This is no different for invariant theory, see for example \cite{DK,Wehlau}. We also point to \cite[Proposition~3.3]{DM-exp} for an elementary linear algebraic description of the invariant ring for torus actions. Conjecture~\ref{con:generators} already fails in this well behaved setting.
This is the content of our Theorem~\ref{thm:torus-action}, which we will prove in this section.
Recall that $\ST_n(\C) \cong \{ \bz \in \T_n : z_1\cdots z_n = 1 \}$, which is itself a torus.

\begin{theorem}[Theorem~\ref{thm:torus-action}, restated]
	Consider the natural action of $G = \ST_n(\C) \times \ST_n(\C) \times \ST_n(\C)$
	on $V = \C^n \otimes \C^n \otimes \C^n$, where an element
	$(\ba, \bb, \bc) \in G$ acts on a tensor $u \in V$ as follows:
  $(\ba,\bb,\bc) \cdot u := v$, such that
	$v_{ijk} = a_i b_j c_k u_{ijk}$. Any set of generators
	for the invariant ring of this action cannot have a polynomial sized
	(in $n$) succinct encoding, unless ${\rm NP \subseteq P/poly}$.
\end{theorem}

\begin{proof}
	Suppose that the natural action above has a set of generators with
	a polynomial sized succinct encoding. Thus, there is an arithmetic
	circuit $\cC(\bx, \by)$ of size $s = {\rm poly}(n)$, where
	$\bx = (x_{ijk})_{i,j,k=1}^n$ is the set of variables corresponding
	to $V$ and $\by = (y_1, \ldots, y_r)$ is the set of auxiliary
	variables, with $r = {\rm poly}(n)$.
	Moreover, from the definition of the size of a succinct encoding we also have that the
	constants used in the computation of $\cC(\bx, \by)$ are rational numbers with bit
	complexity bounded by $b = {\rm poly}(n)$. In particular, $\cC(\bx, \by) \in \Q[\bx,\by]$.

	Let us consider the circuit $\cC(\bx, \by)$ as a circuit whose constants
	are in $\Q[\by]$ and whose variables are only the $\bx$ variables,
	that is, a circuit in $\Q[\by][\bx]$.
  Then, Proposition~\ref{prop:hom-components}
	tells us that there exists a homogeneous circuit $\cC_n(\bx, \by)$,
	in the $\bx$ variables, of degree $n$ and size $O(n^2 s)$ that computes the homogeneous component of $\cC(\bx,\by)$ of degree $n$ as a function of~$\bx$.
  Moreover, the constants of this circuit are a subset of the constants used in the circuit $\cC(\bx, \by)$.
  Since we consider the latter as a circuit in only the $\bx$ variables, the constants in this case are given by the elements of $\Q$ used in the computation of
  $\cC$ as well as the auxiliary variables $\by$.
  In particular, $\cC_n(\bx, \by)$ can be written in the following way:
	\begin{equation}\label{eq:deg-n-cks-1}
		\cC_n(\bx, \by) = \sum_{\bbm \in \cN_n(\bx)} f_\bbm(\by) \cdot \bbm,
	\end{equation}
	where $\cN_n(\bx)$ is the set of all monomials of degree $n$ in the
	variables $\bx$ and $f_\bbm(\by)$ are polynomials in the variables $\by$ of
	degree at most~$2^s$, as the circuit $\cC$ has size at most $s$.

	In Proposition~\ref{prop:least-degree-torus} below, we will show that the invariants
	of minimum degree of our action are in degree $n$, and these
	are spanned by the (maximum) $3$-dimensional matching monomials.
  Thus, if a monomial of degree $n$ is invariant under our action, it must be the case that this monomial corresponds to a $3$-dimensional matching.
  Moreover, the action maps any monomial (invariant or not) to a constant times itself. 
  As $\cC_n(\bx, \by)$ must only compute invariant polynomials, this
	implies that equation~\eqref{eq:deg-n-cks-1} is actually of the following form:
	\begin{equation}\label{eq:deg-n-cks-2}
		\cC_n(\bx, \by) = \sum_{\bbm \in \cM_n(\bx)}
		f_\bbm(\by) \cdot \bbm,
	\end{equation}
	where $\cM_n(\bx)$ is the set of all $3$-dimensional matching monomials
	in the variables $\bx$.
  Moverover, since $\cC(\bx,\by)$ succinctly encodes of a set of generators, the span of $\{\cC_n(\bx,\alpha)\}_{\alpha\in\C^r}$ must necessarily be the same as the span of the 3-dimensional matching polynomials.

	We will now show that the existence of the circuit $\cC_n(\bx, \by)$ implies that
	${\rm NP \in P/poly}$. For that purpose, we will show that given $\cC_n(\bx, \by)$ one can
	solve the $3$-dimensional matching problem in ${\rm P/poly}$.
	Let $H$ be a tripartite 3-uniform hypergraph, whose edges are given by a subset
	$E \subseteq [n] \times [n] \times [n]$.
	We can associate to this graph the tensor $v \in V$ where
	$v_{ijk} = 1$ if hyperedge $(i,j,k) \in E$ and $v_{ijk} = 0$ otherwise.
	Note that $H$ has a $3$-dimensional matching of size $n$ if and only if at least one
	of the $3$-dimensional matching monomials \emph{does not vanish} on our tensor $v$.
	This last
	condition is equivalent to the fact that the circuit $\cC_n(v, \by)$ does not compute the
	zero polynomial (as we know that the span of the set $\{\cC_n(\bx, \alpha)\}_{\alpha \in \C^r}$
	is the same as the span of all $3$-dimensional matching monomials).
	Thus, to solve the $3$-dimensional
	matching problem in ${\rm P/poly}$ it is enough to give a randomized polynomial time algorithm
	for testing whether $\cC_n(v, \by)$ is the zero polynomial or not.\footnote{It is enough to
	give a randomized polynomial time algorithm because we know that ${\rm BPP/poly = P/poly}$.}

	Since $\cC_n(v, \by)$ is a circuit of size ${\rm poly}(n)$ with rational constants of bit
	complexity ${\rm poly}(n)$, it computes a polynomial $P(\by)$ with rational coefficients
	having bit complexity at most $\exp(\poly(n))$ and degree at most $\exp(\poly(n))$.
  This is the setting in which
	Theorem~\ref{thm:PIT-BPP} applies, giving us the desired randomized polynomial
	time algorithm. This concludes our proof modulo Proposition~\ref{prop:least-degree-torus},
	which we will now turn our attention to.
\end{proof}




\noindent
In the following proposition, we denote by $\mathcal{S}_n$ the symmetric group on $n$ letters.

\begin{proposition}\label{prop:least-degree-torus}
	The maximum $3$-dimensional matching monomials  $\prod_{i=1}^n x_{i \sigma(i) \tau(i)}$,
	where $\sigma, \tau \in \mathcal{S}_n$, span the invariants of degree
	$n$ of the natural action of $G = \ST_n(\C) \times \ST_n(\C) \times \ST_n(\C)$
	on $V = \C^n \otimes \C^n \otimes \C^n$. Moreover, there are no nonconstant invariants of degree less than $n$ for this action.
\end{proposition}

\begin{proof}
  Since the action maps any monomial to a constant times itself, it is easy to see that the invariant polynomials are generated by invariant monomials.
	To prove the proposition, it is therefore enough to show that the matching monomials are invariant, that there are no other invariant monomials of degree~$n$, and that there are no invariant monomials of smaller degree.

  We first prove that the matching monomials are invariant.
	Note that the natural action of $G$ on $V$ induces the following action on the variables $x_{ijk}$:
	$(\ba, \bb, \bc) \cdot x_{ijk} = (a_i b_j c_k)^{-1} \cdot x_{ijk}$.%
  \footnote{The inverse comes from the general formula $(g\cdot p)(v) := p(g^{-1} \cdot v)$. }
  Additionally, note that
	$\prod_{\ell=1}^n a_\ell = \prod_{\ell=1}^n b_\ell = \prod_{\ell=1}^n c_\ell = 1$. Given a
	matching monomial $\prod_{i =1}^n x_{i \sigma(i) \tau(i)}$, we therefore have that
	\begin{align*}
		(\ba, \bb, \bc) \cdot \prod_{i =1}^n x_{i \sigma(i) \tau(i)}
		& = \prod_{i =1}^n \left( (a_i b_{\sigma(i)} c_{\tau(i))})^{-1} \cdot x_{i \sigma(i) \tau(i)}  \right) \\
		& = \prod_{i=1}^n (a_i b_{\sigma(i)} c_{\tau(i)})^{-1}  \cdot
		\prod_{i =1}^n x_{i \sigma(i) \tau(i)} \\
		& = \prod_{i =1}^n x_{i \sigma(i) \tau(i)}
	\end{align*}
	where in the last equality we note that for any permutation $\sigma \in \mathcal{S}_n$ (or $\tau$)
	we have $1= \prod_{\ell=1}^n a_\ell = \prod_{\ell=1}^n a_{\sigma(\ell)}$ (and similarly for
	$\bb$ and $\bc$). This proves that the matching monomials are invariant monomials of the natural
	$G$-action on $V$.

	Now, let us prove that no other monomial of degree $n$ is an invariant for this action.
	Let $\prod_{m=1}^n x_{i_m j_m k_m}$ be a monomial, where $(i_m, j_m, k_m) \in [n]^3$, that is not a matching monomial.
	Then there exists some coordinate, w.l.o.g.\ the first coordinate, for which the set~$\{i_m\}_{m=1}^n$ is a strict subset of~$[n]$.
  Equivalently, there is an element $\ell \in [n]$ such that $\ell \not\in \{i_m\}_{m=1}^n$.
  W.l.o.g., we can assume that
	$\ell = 1$. Thus, the action of $\ba = (\alpha^{n-1}, \alpha^{-1}, \ldots, \alpha^{-1}),
	\bb = \bc = (1, \ldots, 1)$  on our monomial $\prod_{m=1}^n x_{i_m j_m k_m}$ is as follows:
	\begin{equation*}
		(\ba, \bb, \bc) \cdot \prod_{m=1}^n x_{i_m j_m k_m}
		= \alpha^n \cdot \prod_{m=1}^n x_{i_m j_m k_m}
	\end{equation*}
	which proves that this monomial is not an invariant.
  This completes the proof that the matching
	monomials span the invariants of degree $n$.

	Now we are left with proving that there are no nonconstant monomials of degree less than~$n$ that are invariant.
  Note that if we have a monomial with degree less than $n$, we can represent it as
	$\prod_{m=1}^d x_{i_m j_m k_m}$, where $d < n$ and by the pigeonhole principle, we know that
	there exists $\ell \in [n]$ such that $\ell$ does not appear as a first coordinate entry in
	the set of tuples $\{(i_m, j_m, k_m)\}$. If $d>0$ then, analogously to the previous paragraph, we know
	that such monomials cannot be invariants of the natural action of $G$ over $V$, therefore showing
	that no nonconstant monomial of degree $< n$ can be an invariant. This completes the proof.
\end{proof}

\section{Invariant Theory for \texorpdfstring{$\SL_n(\C)$}{SL\_n(C)} and Mulmuley's conjecture} \label{Sec:Ketan}

\input{ketan}

\bibliographystyle{alpha}
\bibliography{refs}

\end{document}

%% file: prelim.tex

In this section we establish notation and we formally state basic facts and definitions which
we will need in later sections.

\begin{definition}[3-dimensional matching~\cite{K72}]\label{def:3D-matching}
	The \emph{3-dimensional matching} problem is defined as follows:
	\begin{itemize}
		\item[Input:] a set $U \subseteq [n] \times [n] \times [n]$, representing the edges of
		a tripartite, 3-uniform hypergraph.
		\item[Output:] YES, if there is a set of hyperedges $W \subseteq U$ such that $|W| = n$ and
		no two elements of $W$ agree in any coordinate (that is, they form a matching in this
		hypergraph). NO, if there is no such set.
	\end{itemize}
\end{definition}

\begin{theorem}[NP-completeness of 3-dimensional matching~\cite{K72}]\label{thm:3D-matching-NP-complete}
	The 3-dimensional matching problem is NP-complete.
\end{theorem}

\subsection{Basic facts from algebraic complexity}

We now give basic facts that from algebraic complexity which we will use in the next sections.

The next proposition shows that homogeneous components of low degree
of an arithmetic circuit can be efficiently computed, with a small blow-up
in circuit size and without the use of any extra constants. This proposition was
originally proved by Strassen in~\cite{S73} and its proof can be found
in~\cite[Theorem 2.2]{SY10}.
In the following proposition, given a polynomial $p(\bx)$, we denote its degree-$d$ homogeneous component by $H_d[p(\bx)]$.

\begin{proposition}[Efficient computation of homogeneous components]\label{prop:hom-components}
	Given a circuit $\cC(\bx)$ of size $s$, then for every $r \in \N$
	there is a homogeneous circuit $\Psi(\bx)$ of size $O(r^2 s)$
	computing $H_0[\cC(\bx)], H_1[\cC(\bx)], \ldots, H_r[\cC(\bx)]$.
	Moreover, the constants used in the computation of the components $H_i[\cC(\bx)]$
	are a subset of the coefficients used in the computation of $\cC(\bx)$.
\end{proposition}

The next theorem, proved by~\cite[Theorem 4.10]{AB03} gives us a randomized polynomial time
algorithm to test whether an algebraic circuit of polynomial size, with rational
coefficients, is identically zero. Another randomized algorithm easily follows
from~\cite[Lemma 2]{S79}, when adapted for polynomials with rational coefficients.

\begin{theorem}[PIT for poly-sized circuits~\cite{AB03}]\label{thm:PIT-BPP}
	Let $P(\bx) \in \Q[\bx]$ be a polynomial in the variables $\bx = (x_1, \ldots, x_n)$, with
	each variable $x_i$ having degree bounded by $d_i$, and whose coefficients are rationals with
	bit complexity bounded by $B$. If $P(\bx)$ is given as an arithmetic circuit of size $s$, then
	there exists a randomized algorithm running in time ${\rm poly}(n, s, \log(B), 1/\epsilon)$
	and using $O\left(\sum_{i=1}^n \log(d_i) + \log(B) \right)$ random bits which tests whether $P(\bx)$ is identically zero.
	If $P(\bx)$ is the identically zero polynomial then the algorithm always succeeds.
	Otherwise, it errs with probability at most $\epsilon$.
\end{theorem}

%% file: ketan.tex

In this section, we will give another example of a group action on tensors for which any set of generating invariants is hard to compute, i.e., we will prove Theorem~\ref{MC-cex}. Even though the previous section already gives a counterexample, this example illustrates something more. The feature of this group action is that invariants of minimial degree span a $1$-dimensional space. In other words, up to scaling, we have a unique invariant of minimal degree. This unique invariant in the minimal degree is called the \emph{hyperpfaffian polynomial} (introduced by Barvinok in 1995 as a natural generalization of the well-known Pfaffian polynomial to higher order tensors). We then study the hyperpfaffian's computational complexity and prove that it is VNP-complete. The importance of this example is that such a unique invariant in the minimal degree is {\em essential} in any generating set. So, it is not even possible to give a generating set consisting of invariant polynomials that are easy to compute, even if we remove all restrictions on the size of the generating set.
Moreover, the group action is by $\SL_n(\C)$ rather than a torus.
Therefore our counterexample disproves Mulmuley's original conjecture in a strong sense.


\subsection{Invariant Rings and Symmetric Tensors}\label{sec:invariants}
The special linear group $\SL_n(\C)$ consists of complex $n\times n$-matrices with unit determinant and acts canonically on $\C^n$ by matrix-vector multiplication.
This action extends to any $m$-th tensor power $\tensor^m \C^n$ by
\begin{equation}\label{eq:SL on tensor power}
g \cdot (v_1 \otimes \cdots \otimes v_m) := (g \cdot v_1) \otimes \cdots \otimes (g \cdot v_m)
\end{equation}
and linear continuation.
We will always use the standard bilinear form on $\tensor^m \C^n$ that satisfies $\langle g^T \cdot v, w\rangle = \langle v, g \cdot w \rangle$ for all $v,w \in \tensor^m\C^n$, $g \in \SL_n(\C)$.

Let~$V$ be an arbitrary finite-dimensional $\SL_n(\C)$-representation (such as $V = \tensor^m\C^n$).
Then $\SL_n(\C)$ also acts on~$\C[V]_d$, the vector space of degree-$d$ homogeneous polynomials on~$V$, via the formula
\[
(g \cdot p)(v) := p(g^T \cdot v),
\]
where $p \in \C[V]_d, g \in G$ and $v \in V$.
The formula above is the dual representation of the action on the ring of all polynomial functions~$\C[V] = \bigoplus_{d=0}^\infty \C[V]_d$ that we explained in in Section~\ref{subsec:inv theo}.
Using the dual here is only for presentation purposes, as it gives a clearer connection to multilinear algebra as follows.
Note that a polynomial $p \in \C[V]$ is invariant if and only if $\forall g \in \SL_n(\C)$ we have $g \cdot p=p$.

It is convenient to identify polynomial functions with symmetric tensors.
Note that $\SL_n(\C)$ acts canonically on any $d$-th tensor power $\tensor^d V$ of~$V$.
This action restricts to the $d$-th symmetric tensor power~$\Sym^d V$, i.e., the $\mathcal{S}_d$-invariant subspace of $\tensor^d V$. Recall that $\mathcal{S}_d$ is the symmetric group on $d$ letters; it acts on $V^{\otimes d}$ by permuting tensor factors.
For any~$t \in \Sym^d V$, we can define a homogeneous degree-$d$ polynomial~$p\in\C[V]_d$ by~$p(v) := \langle t, v^{\otimes d}\rangle$.
Here we use the quadratic form on $\Sym^d V$ induced by a non-degenerate bilinear form on $V$ that satisfies $\langle g^t\cdot v, w\rangle = \langle v,g w\rangle$ for all $v,w \in V$, $g \in \SL_n(\C)$ as above.
Then, $p$ is invariant if and only if the symmetric tensor~$t$ is invariant, i.e., if $\forall g \in \SL_n(\C)$ we have~$g \cdot t=t$.
We will tacitly go back and forth between symmetric tensors in $\Sym^d \tensor^m \C^n$ and homogeneous polynomials in~$\C[\tensor^m \C^n]_d$.

Now, we turn to studying hyperpfaffians.



\subsection{Hyperpfaffians}
The \emph{Pfaffian} is the unique (up to scale) homogeneous $\SL_{2n}(\C)$-invariant of degree $n$ on~$\C^{2n}\otimes\C^{2n}$.
There are no $\SL_{2n}(\C)$-invariants in lower degrees.
If we identify $\C^{2n} \otimes \C^{2n}$ with the space of complex $2n\times 2n$ matrices~$A$, then the Pfaffian is invariant under the action of $\SL_{2n}(\C)$ given by $g \cdot A := gAg^T$.
The defining property of the Pfaffian generalizes to tensors of even order as follows (the classical Pfaffian is the special case of~$k=1$):

\begin{proposition}\label{pro:nolowerdegreeinv}
For any $k$ and $n$, there is a unique (up to scale) homogeneous $\SL_{2kn}(\C)$-in\-va\-ri\-ant polynomial~$\Pf_{k,n}$ of degree~$n$ on $\tensor^{2k} \C^{2kn}$.
$\Pf_{k,n}$ identifies with the symmetric tensor $e_1 \wedge \cdots \wedge e_{2kn} \in \Sym^n \tensor^{2k}\C^{2kn}$.
There are no nonconstant $\SL_{2kn}(\C)$-invariants of lower degree.
\end{proposition}

\noindent
Before proving Proposition~\ref{pro:nolowerdegreeinv} we recall some representation theory.
The material is well-known, and we refer to standard texts (e.g., \cite{FH}) for details.
A \emph{partition} $\la$ is a nonincreasing sequence of natural numbers with finite support.
We write $\la \vdash_n m$ to say that~$|\la| := \sum_i \la_i = m$ and $\la_{n+1}=0$.
If $\la_{n+1}=0$, then we say that $\la$ is an $n$-partition.
The irreducible polynomial $\GL_n(\C)$-representations are indexed by $n$-partitions.
For a partition~$\la$, let~$\{\la\}$ denote the irreducible $\GL_n(\C)$-representation corresponding to~$\la$.
Restricted to $\SL_n(\C)$, the representation $\{\la\}$ is trivial if and only if $\la_1=\ldots=\la_n$; note that this implies that $n \mid m$.
The irreducible representations of $\mathcal{S}_m$ are indexed by partitions~$\la$ with~$|\la|=m$.
Let $[\la]$ denote the irreducible $\mathcal{S}_m$-representation corresponding to~$\la$.

Consider $\tensor^m \C^n$. This space has an action of $\SL_n(\C)$ by~\eqref{eq:SL on tensor power}, but also an action of $\mathcal{S}_m$ that permutes the tensor factors.
Both actions commute, so we have an action of the product group $\SL_n(\C)\times\mathcal S_m$. The following well-known result will be crucial for our purposes.

\begin{theorem} [Schur--Weyl duality]
As an $\SL_n(\C) \times \mathcal{S}_m$-representation, we have the decomposition:
\[ 
\tensor^m \C^n = \bigoplus_{\la \vdash_n m} \{\la\}\otimes[\la].
\] 
\end{theorem}

\noindent
Using Schur--Weyl duality, one sees immediately that $\tensor^m \C^n$ contains nonzero $\SL_n(\C)$-invariant vectors if and only if $n \mid m$.
This is because a vector is invariant if and only if it spans a trivial irreducible representation -- but $\{\la\}$ is trivial if and only if $\la_1=\ldots=\la_n$, as mentioned above.
For $m=n$, there is a unique (up to scale) $\SL_n(\C)$-invariant vector.
This is because the invariants in $\tensor^n \C^n$ correspond to the component $\{1^n\} \otimes [1^n]$, where we write~$1^n$ for the partition~$\la_1=\ldots=\la_n=1$.
Here, $\{1^n\}$ is the trivial representation of $\SL_n(\C)$ and $[1^n]$ is one-dimensional, as it is the sign representation of~$\mathcal{S}_n$.
Thus the space of invariants is one-dimensional. This unique vector (up to scale) is given by the wedge product $e_1 \wedge e_2 \wedge \cdots \wedge e_n$, where $a \wedge b := \frac 1 2 (a \otimes b - b \otimes a)$, and higher order wedge products are defined analogously.

\begin{proof}[Proof of Proposition~\ref{pro:nolowerdegreeinv}]
It suffices to show that $\Sym^d \tensor^{2k} \C^{2kn}$ contains no $\SL_{2kn}(\C)$-invariant vector if $0<d<n$ and that it contains a unique such vector if $d=n$.
Note that $\Sym^d \tensor^{2k} \C^{2kn}$ is a subspace of $\tensor^d \tensor^{2k} \C^{2kn} \simeq \tensor^{2kd} \C^{2kn}$.
Thus the first claim holds since $\tensor^{2kd} \C^{2kn}$ contains $\SL_{2kn}(\C)$-invariant vectors only if $2kn \mid 2kd$. Thus if $0 < d < n$, there are no invariants.
For~$d=n$, $\tensor^d \tensor^{2k} \C^{2kn} \simeq \tensor^{2kn} \C^{2kn}$ contains the unique $\SL_{2kn}(\C)$-invariant vector $v = (e_1 \wedge \cdots \wedge e_{2k}) \wedge \cdots \wedge (e_{2k(n-1)+1} \wedge \cdots \wedge e_{2kn})$.
It remains to show that $v$ is symmetric, i.e., an element of~$\Sym^d \tensor^{2k} \C^{2kn}$, which is a subspace of $\tensor^d \tensor^{2k} \C^{2kn}$.
But this is easy to see since each of the~$d$ blocks has even size~$2k$ and the wedge product is skew-commutative.
This proves the second claim.
\end{proof}

\noindent
The polynomial~$\Pf_{k,n}$ was introduced in~\cite[Def.~3.4]{Bar:95} in its monomial presentation, where it is called the \emph{hyperpfaffian}.
Note that, for fixed $k$, $\Pf_k := (\Pf_{k,1},\Pf_{k,2},\dots)$ is a p-family (i.e., both the degree and the number of variables are polynomially bounded), since~$\Pf_{k,n}$ has degree~$n$ and $(2kn)^{2k}$ variables.
The monomial presentation in~\cite{Bar:95} immediately yields that $\Pf_k \in \VNP$.
\begin{theorem}
For even $k$, $\Pf_k$ is \VNP-complete.
\end{theorem}
\begin{proof}
We present a projection of $\Pf_{k,d}$ to the $d\times d$ permanent.
The same projection yields the determinant if $k$ is odd, which explains why the proof does not work for the classical Pfaffian ($k=1$).
The case~$k=2$ is enough to disprove Mulmuley's conjecture.

By Proposition~\ref{pro:nolowerdegreeinv}, the Pfaffian $\Pf_{k,d}$ identifies with the symmetric tensor
\begin{align*}
  v := e_1 \wedge \cdots \wedge e_{2kd} \in \Sym^d \tensor^{2k}\C^{2kd}.
\end{align*}
Thus, the evaluation~$\Pf_{k,d}(p)$ at a tensor $p \in \tensor^{2k} \C^{2kd}$ is given by $\langle v, p^{\otimes d}\rangle$ (cf.~\cite[Sec.~4.2(A)]{Ike:12}).
We choose
\[
p = \sum_{i,j=0}^{d-1} x_{i+1,j+1} (e_{1+2ki} \otimes e_{2+2ki} \otimes \cdots \otimes e_{k+2ki} \otimes e_{k+1+2kj} \otimes e_{k+2+2kj} \otimes \cdots \otimes e_{2k+2kj}),
\]
where the $x_{i,j}$ ($1 \leq i,j \leq d$) are formal variables.

The point $p$ is parametrized linearly by the $x_{i,j}$, so the evaluation of $\Pf_{k,d}$ at $p$ is a projection of $\Pf_{k,d}$.
We verify that the evaluation of $\Pf_{k,n}$ at $p$ gives the $d \times d$ permanent (up to a constant nonzero scalar) as follows.

{
\fontsize{10}{12}
\selectfont
\begin{align*}
\quad p^{\otimes d} &=\hspace{-.5cm}\sum_{i_1,j_1,\ldots,i_d,j_d=0}^{d-1}\hspace{-.5cm} x_{i_1+1,j_1+1} \cdots x_{i_d+1,j_d+1}
(e_{1+2ki_1} \otimes e_{2+2ki_1} \otimes \cdots \otimes e_{k+2ki_1} \otimes e_{k+1+2kj_1} \otimes e_{k+2+2kj_1} \otimes \cdots \otimes e_{2k+2kj_1})\\ &\quad\quad\quad\quad\quad\quad\quad\quad\otimes \cdots \otimes
(e_{1+2ki_d} \otimes e_{2+2ki_d} \otimes \cdots \otimes e_{k+2ki_d} \otimes e_{k+1+2kj_d} \otimes e_{k+2+2kj_d} \otimes \cdots \otimes e_{2k+2kj_d})
\end{align*}
}
and by linearity
{
\fontsize{10}{12}
\selectfont
\begin{align*}
\quad \langle v,p^{\otimes d}\rangle &= \hspace{-.5cm}\sum_{i_1,j_1,\ldots,i_d,j_d=0}^{d-1}\hspace{-.5cm} x_{i_1+1,j_1+1} \cdots x_{i_d+1,j_d+1}
\langle v,(e_{1+2ki_1} \otimes e_{2+2ki_1} \otimes \cdots \otimes e_{k+2ki_1} \otimes e_{k+1+2kj_1} \otimes e_{k+2+2kj_1} \otimes \cdots \otimes e_{2k+2kj_1})\\ &\quad\quad\quad\quad\quad\quad\quad\quad\otimes \cdots \otimes
(e_{1+2ki_d} \otimes e_{2+2ki_d} \otimes \cdots \otimes e_{k+2ki_d} \otimes e_{k+1+2kj_d} \otimes e_{k+2+2kj_d} \otimes \cdots \otimes e_{2k+2kj_d})\rangle
\end{align*}
}
A crucial property of $v$ is that $\langle v,e_{\pi(1)} \otimes e_{\pi(2)} \otimes \cdots \otimes e_{\pi(n)}\rangle \neq 0$ iff $\pi$ is a permutation, in which case it is equal to the sign of the permutation.
It follows that the nonzero summands in $\langle v,p^{\otimes d}\rangle$ are precisely those for which~$i=(i_1,\ldots,i_d)$ and $j=(j_1,\ldots,j_d)$ are permutations of $\{0,\ldots,d-1\}$.
For a single summand with $i$ and $j$ permutations we see:
\begin{align*}
&\quad x_{i_1+1,j_1+1} \cdots x_{i_d+1,j_d+1}
\langle v,(e_{1+2ki_1} \otimes e_{2+2ki_1} \otimes \cdots \otimes e_{k+2ki_1} \otimes e_{k+1+2kj_1} \otimes e_{k+2+2kj_1} \otimes \cdots \otimes e_{2k+2kj_1})\\ &\quad\quad\quad\quad\quad\quad\quad\quad\otimes \cdots \otimes
(e_{1+2ki_d} \otimes e_{2+2ki_d} \otimes \cdots \otimes e_{k+2ki_d} \otimes e_{k+1+2kj_d} \otimes e_{k+2+2kj_d} \otimes \cdots \otimes e_{2k+2kj_d})\rangle
\\&= \sgn(i)^k \sgn(j)^k x_{i_1+1,j_1+1} \cdots x_{i_d+1,j_d+1}.
\end{align*}
Hence, for even $k$ we obtain $\langle v,p^{\otimes d}\rangle = d! \per_d$.
\end{proof}

\noindent
Finally, we put together the preceding results to prove Theorem~\ref{MC-cex}.

\begin{theorem}[Theorem~\ref{MC-cex}, restated]
Let $k \geq 2$ be even.
Consider the action of $G = \SL_{2kn}(\C)$ on $V = \tensor^{2k} \C^{2kn}$. Then any set of generators for the invariant ring cannot have a polynomial sized (in $n$) succinct encoding, unless ${\rm VP = VNP}$.
\end{theorem}
\begin{proof}

We summarize the results so far. Let $k \geq 2$. Consider the action of $G = \SL_{2kn}(\C)$ on $V = \tensor^{2k} \C^{2kn}$. Then:

\begin{enumerate}
\item There are no homogeneous invariant polynomials of degree $< n$.

\item The space of homogeneous invariant polynomials of degree $n$ is $1$-dimensional, and spanned by the hyperpfaffian polynomial $\Pf_{k,n}$.

\item The hyperpfaffian polynomial $\Pf_{k,n}$ is VNP-complete.
\end{enumerate}
The rest of the proof proceeds along the same lines as in the proof of Theorem~\ref{thm:torus-action} in Section~\ref{sec:torus}. If we had a poly-sized succinct encoding for the generators of this invariant ring, then one would be able to extract the lowest degree part, which would yield a poly-sized circuit computing $\Pf_{k,n}$.
This is not possible unless ${\rm VP = VNP}$, since $\Pf_{k,n}$ is VNP-complete.
\end{proof}